\documentclass[letterpaper, 10 pt, conference]{IEEEtran}
\IEEEoverridecommandlockouts

\usepackage{geometry}
\usepackage{graphics} 
\usepackage{graphicx}

\usepackage{amsmath, amsfonts, amssymb, amsthm} 
\usepackage[mathscr]{euscript}
\usepackage{xcolor} 
\usepackage{etoolbox}

\usepackage{caption}
\usepackage{subcaption}
\let\labelindent\relax
\usepackage[shortlabels]{enumitem}

\newtoggle{sepsections}
\togglefalse{sepsections}  

\newcommand{\cp}{\iftoggle{sepsections}{\clearpage}{}}
\renewcommand{\baselinestretch}{0.95}

\newcommand{\mb}[1]{\mathbf{#1}}
\newcommand{\norm}[1]{\left\lVert #1 \right\rVert}

\newcommand{\E}[2][]{\mathbb{E}_{#1}\left[#2\right]}

\newcommand{\R}{\mathbb{R}}
\newcommand{\partialto}{\rightharpoonup}

\DeclareMathOperator*{\argmax}{argmax}
\renewcommand{\argmax}{\operatornamewithlimits{argmax}}

\newtheorem{theorem}{Theorem}
\newtheorem{lemma}{Lemma}

\theoremstyle{definition}
\newtheorem{definition}{Definition}
\newtheorem{corollary}{Corollary}
\newtheorem*{remark}{Remark}

\newcommand{\newsec}[1]{\vspace{2mm} \noindent \textbf{#1.} }

\newcommand{\dtwo}{\eta}
\newcommand{\ellp}{{}}

\usepackage{hyperref} 

\title{\LARGE \bf Input-to-State Stability in Probability}

\author{Preston Culbertson$^{1}$, Ryan K. Cosner$^{1}$, Maegan Tucker$^{1}$, and Aaron D. Ames$^{1,2}$
\thanks{This research was supported by the National Science Foundation (CPS Award \#1932091), BP, and the Zeitlin Family Fund}%
\thanks{$^{1}$ Authors are with the Department
of Mechanical and Civil Engineering, California Institute of Technology,
Pasadena, CA 91125, USA; \texttt{\{pculbert, rkcosner, mtucker, ames\}@caltech.edu}}.%
\thanks{$^{2}$ Authors are with the Department of Control and Dynamical Systems, California Institute of Technology,
Pasadena, CA 91125, USA.}

}

\begin{document}
\maketitle

\begin{abstract}

Input-to-State Stability (ISS) is fundamental in mathematically quantifying how stability degrades in the presence of bounded disturbances. If a system is ISS, its trajectories will remain bounded, and will converge to a neighborhood of an equilibrium of the undisturbed system. This graceful degradation of stability in the presence of disturbances describes a variety of real-world control implementations.  Despite its utility, this property requires the disturbance to be bounded and provides invariance and stability guarantees only with respect to this worst-case bound. In this work, we introduce the concept of ``ISS in probability (ISSp)'' which generalizes ISS to discrete-time systems subject to unbounded stochastic disturbances. Using tools from martingale theory, we provide Lyapunov conditions for a system to be exponentially ISSp, and connect ISSp to stochastic stability conditions found in literature. We exemplify the utility of this method through its application to a bipedal robot confronted with step heights sampled from a truncated Gaussian distribution. 


\end{abstract}

\section{Introduction}


Control systems operating in practice are nearly always affected by disturbances, be they noise, modelling error, uncertain state estimates, or environmental interactions.  This motivates the design of controllers which are robust to these uncertainties. Input-to-state stability (ISS) \cite{sontag_smooth_1989, sontag_characterizations_1995} is a useful heuristic for the robustness of a control system. If a system is ISS, then, loosely, when the system is subjected to bounded disturbances, the system state will converge to some ball whose radius scales with the maximum disturbance norm; in the presence of zero disturbances, the system is asymptotically stable. ISS can be interpreted as guaranteeing the ``graceful degredation'' of asymptotic stability under bounded disturbances; bounded disturbance inputs still produce bounded state trajectories, and asymptotic stability is recovered as the input magnitude approaches zero.

However, as a robustness property, ISS suffers some drawbacks, particularly when reasoning about systems subject to stochastic disturbances. The central issue is that ISS reasons only about bounded disturbances, i.e., those whose norm is upper-bounded. However, many noise sources are more naturally modeled as continuous, unbounded random variables (e.g., systems subject to additive Gaussian noise); ISS-based tools cannot handle such cases. Further, the guarantees provided by ISS depend on the worst-case disturbance magnitude and are thus often highly conservative. 

To remedy these limitations, in this paper we introduce a generalization of the ISS property for discrete-time systems subject to unbounded stochastic disturbances: \textit{input-to-state stability in probability (ISSp)}. Intuitively, we say a system is ISSp if the typical ISS condition can hold with a probability arbitrarily close to one for an arbitrary (but finite) horizon. Using tools from martingale theory, we provide Lyapunov conditions for the exponential form of ISSp. We also explore connections between ISS, ISSp, and more traditional stability notions for stochastic systems. We conclude with simulation studies of ISSp systems subject to unbounded disturbances, including a double-integrator subject to additive Gaussian noise and a bipedal robot walking on uncertain terrain as illustrated in Fig. \ref{fig:walker-config}. In particular, we show that our ISSp-based exit probability bound is indeed conservative for all examples, and show that we can provide non-trivial probabilistic stability guarantees for the biped over a larger disturbance set than the worst-case ISS bound \cite{tucker_input--state_2023}.

\begin{figure}[tb]
    \centering
    \includegraphics[width=0.9\linewidth]{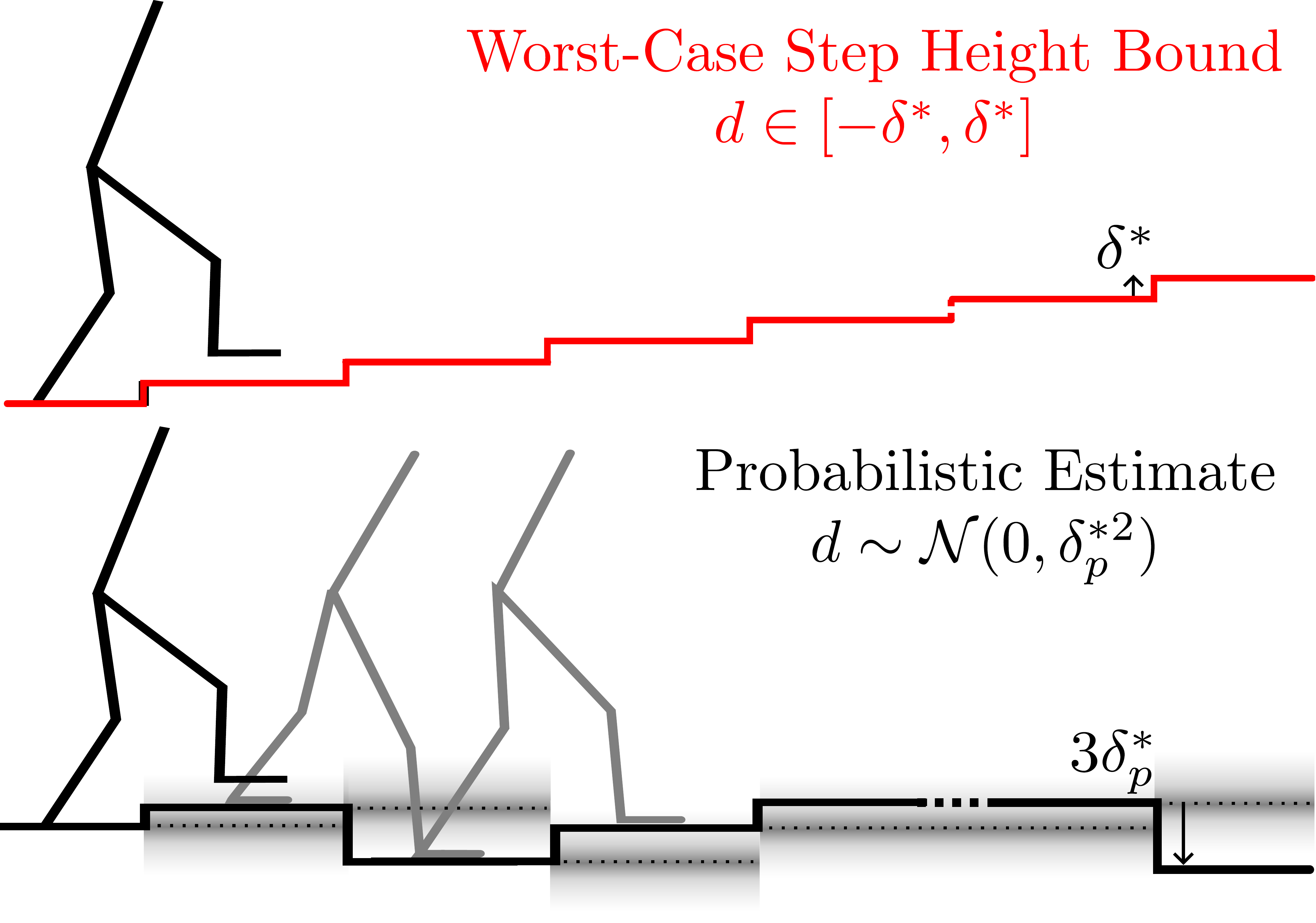}
    \caption{This paper introduces input-to-state stability in probability (ISSp),  which generalizes input-to-state stability (ISS) to systems with unbounded disturbances. We use this framework to study a seven-link walker traversing stochastic terrain.  While ISSp only provides probabilistic guarantees, we find our framework yields more reasonable estimates for the tolerable set of step heights. \textbf{(Top):} ISS-based guarantees must hold for any (bounded) disturbance signal; even for worst-case terrain (e.g., stairs) the walker must be able to remain stable. \textbf{(Bottom):} ISSp reasons instead about how systems behave over finite horizons. While the stochastic step heights (shown as gradients) can sometimes be large, their distribution is concentrated near zero, and thus the walker has a high probability of remaining upright.}
    \label{fig:walker-config}
    \vspace{-6mm}
\end{figure}


There has been a large body of work on ISS. Beginning with the seminal papers by Sontag \cite{sontag_characterizations_1995, sontag_input--state_1995} 
for continuous time and the extension of these results to discrete time \cite{jiang_input--state_2001}, ISS has found utility in the fields of control theory \cite{nesic_dragan_input--state_nodate} robotics \cite{angeli_characterization_2000}, and of special note to this paper, robotic walking \cite{ma2017bipedal,kolathaya2018input,tucker2023input}.  This paper leverages results on martingale theory \cite{kushner_stochastic_1967} to extend (discrete-time) ISS to stochastic systems.  In particular, the stochastic notions presented in this paper are similar to the set-invariance notions in \cite{steinhardt_finite-time_2012,santoyo_verification_2019}, but differ in that they add a notion of stochastic stability to reflect the convergence present in systems that are ISS. Theoretically, this work uses similar tools to those presented in \cite{cosner_robust_2023}, which uses a similar supermartingale to bound the finite-time exit probability of a system from a particular set. However, in this work, we provide a novel “ISS-like” interpretation of this supermartingale, and make explicit connections between the concepts of ISS and other stochastic stability notions found in literature such as variable drift \cite{lengler2020drift}, recurrence \cite{meyn1993markov}, and boundedness of trajectories in probability \cite{steinhardt_finite-time_2012}.

\cp 
\section{Background}

Consider a discrete-time autonomous system,
\begin{align}
    \mb{x}_{k+1} = \mb{f}(\mb{x}_k, \mb{d}_k) \label{eq:dt_dynamics}
\end{align}
with $k \in \mathbb{N}_{\geq 0}$, state $\mb{x}_k \in \mathcal{X} \subseteq \mathbb{R}^n$, equilibrium point $(\mb{x}^*, \mb{d}^*) = (\mb{0}, \mb{0})$,  random disturbance $\mb{d}_k \in \mathbb{R}^d$, and continuous dynamics $\mb{f}: \mathcal{X} \times \R^d \to \R^n$. We assume each disturbance $\mb{d}_k \overset{\text{i.i.d.}}{\sim} \mathcal{D}$ from some disturbance distribution $\mathcal{D}.$

\subsection{Input-to-State Stability for Deterministic Systems}


If the disturbance distribution for system \eqref{eq:dt_dynamics} is bounded, then we can use the concept of Input-to-State Stability (ISS) to reason about the boundedness and convergence of the system's trajectories. 

\begin{definition}[Input-to-State Stability \cite{jiang_input--state_2001}]
    The system \eqref{eq:dt_dynamics} is input-to-state stable (ISS) if there exist functions\footnote{A continuous function $\gamma:[0, a) \to [0, \infty)$ for $a > 0$ is said to belong to class $\mathcal{K}$  ($\gamma \in \mathcal{K}$) if it is strictly monotonically increasing and $\gamma(0) = 0 $. If additionally $a = \infty$ and $\gamma(r) \to \infty $ as $r \to \infty$ then $\gamma$ belongs to $\mathcal{K}_\infty$. A continuous function $\beta : [0, a) \times [0, \infty)  \to [0, \infty)$ is said to belong to class $\mathcal{KL}$ if for each fixed $s\geq 0$  the function $\beta( \cdot , s) $ is class $\mathcal{K}$ and for each $r \geq 0 $ the function $\beta(r, \cdot) $ is decreasing and $\beta(r, s) \to 0  $ as $s \to \infty$. } $\beta \in \mathcal{KL}$  and $\gamma \in \mathcal{K}$ such that, for each deterministic disturbance input $\mb{d}_k \in \R^m$ and each $\mb{x}_0 \in \R^n $, it holds that
    \begin{align}
        \Vert \mb{x}_k \Vert_{\ellp} \leq \beta (\Vert \mb{x}_0 \Vert_{\ellp}, k) + \gamma \left( \sup_{k \in \mathbb{N}_{\geq 0 }} \Vert \mb{d}_k\Vert_{\ellp} \right) \tag{ISS} \label{eq:ISS}
    \end{align}
    for each $k \in \mathbb{N}_\geq{0}$ and some $p \geq 1  $. 
\end{definition}

Intuitively, the bound on the state trajectory is a function of a sequence which converges to zero in time, $\beta(\Vert \mb{x}_0 \Vert_{\ellp}, k) $ and a term which grows with respect to the disturbance bound, $\gamma \left( \sup_{k \in \mathbb{N}_{\geq 0}} \Vert \mb{d}_k \Vert_{\ellp} \right)$. If $\norm{\mb{d}_k}_{\ellp} = 0$ for all $k$, then ISS systems are asymptotically stable. Note that a similar inequality regarding an essentially bounded disturbance distribution $\mathcal{D}$:
    \begin{align}
        \Vert \mb{x}_k \Vert_{\ellp} \leq \beta (\Vert \mb{x}_0 \Vert_{\ellp}, k) + \gamma \left(\textrm{ess sup} \Vert \mathcal{D}\Vert \right), 
    \end{align}
can be employed to achieve ISS almost surely (Corollary \ref{cor:iss-infty}) where $\textrm{ess sup}$ is the essential supremum of the distribution $\mathcal{D}$, also written as the $L^\infty$-norm of $\mathcal{D}$. 

We now introduce ISS-Lyapunov functions as tools for verifying a system's ISS property.

\begin{definition}[ISS-Lyapunov Function \cite{jiang_input--state_2001}]
    A continuous function $V: \R^n \to \R_{\geq 0} $ is an ISS Lyapunov function for \eqref{eq:dt_dynamics} if there exist $\kappa_1, \kappa_2, \kappa_3 \in \mathcal{K}_\infty $ and $\kappa_4 \in \mathcal{K}$ such that: 
        \begin{align}
            \kappa_1(\Vert \mb{x}\Vert_{\ellp}) \leq V(\mb{x}) & \leq \kappa_2 (\Vert \mb{x} \Vert_{\ellp})\\
            V(\mb{f}(\mb{x}, \mb{d})) - V(\mb{x}) & \leq - \kappa_3(\Vert \mb{x} \Vert_{\ellp}) + \kappa_4 (\Vert \mb{d} \Vert_{\ellp}) \label{eq:iss_lyap}
        \end{align}
        for all $\mb{x} \in \R^n$ and all $\mb{d} \in \R^d$. Additionally, $V$ is an exponential-ISS (E-ISS) Lyapunov function if there exist constants $a, b, c > 0$ and $\alpha \in (0,1) $ such that $\kappa_1(r) = a r^c, \kappa_2(r) = br^c, \kappa_3(r) = \alpha r^c $,
\end{definition}

The existence of an ISS-Lyapunov function can now be immediately used to verify that the system is ISS. 

\begin{theorem}[\cite{jiang_input--state_2001}]
    If there exists an ISS Lyapunov Function for system \eqref{eq:dt_dynamics}, then system \eqref{eq:dt_dynamics} is ISS. 
\end{theorem}

\subsection{Stochastic Preliminaries: $L^p$ spaces and Martingales}

Here we provide a brief discussion of random variables, martingales, and other tools that we will use to generalize ISS to the case of unbounded, stochastic disturbances. We will present this material at a level necessary to communicate these concepts clearly and accessibly. We refer readers to \cite{grimmett_probability_2020} for a precise measure-theoretic presentation of these ideas.

In this paper we consider disturbance signals which are sequences of random variables. A continuous random variable $\mb{y}$ sampled from a distribution $\mathcal{Y}$ (denoted $\mb{y} \sim \mathcal{Y}$) is a quantity that takes on values in $\R^{y}$ according to a probability density $p(\mb{y}) \geq 0$, with $\mathbb{P} \{ \mb{y} \in A \} \triangleq \int_A p(\boldsymbol{\upsilon}) d \boldsymbol{\upsilon} $. By definition $\int_{\R^{y}}p(\boldsymbol{\upsilon}) d \boldsymbol{\upsilon} = 1$ and the expectation of a random variable is given by $\E{\mb{y}} \triangleq \int_{\R^{y}} \boldsymbol{\upsilon} p(\boldsymbol{\upsilon}) d \boldsymbol{\upsilon}$.


We now introduce $L^p$ spaces of random variables.  

\begin{definition}[$L^p$ Space \cite{wheeden_measure_2015}]
    A random variable $\mb{y}\sim \mathcal{Y}$ belongs to $L^p$ (denoted as $\mb{y} \in L^p$), for $p > 0$, if
    \begin{align}
        \norm{\mb{y}}_{L^p} \triangleq \E{\norm{\mb{x}}^p}^\frac{1}{p} < \infty.
    \end{align}
\end{definition}
\noindent We call $\norm{\cdot}_{L^p}$ the $p$-norm of a random variable, which is finite for any random variable in $L^p.$ Intuitively, for $0 < p \leq q$, $L^q \subseteq L^p$ \cite[Thm. 8.2]{wheeden_measure_2015} since random variables in $L^q $ have tails that decay faster than those in $L^p$; additionally, $L^\infty$ is the smallest $L^p$ space and only contains random variables that are essentially bounded. Note that any norm $\norm{\cdot}$ appearing without a subscript defines a typical norm on $\mathbb{R}^n$. 
    
We can also reason about a random variable's conditional probability, i.e., its distribution given that another random variable has taken on a particular value. For two random variables $X, Y$ the density of $X$ given $Y=y$ is given by
\begin{align*}
    p(x \mid y) = \frac{p(x, y)}{p(y)} 
\end{align*}
where $p(x, y)$ is the joint probability density of $X, Y$. The conditional expectation of $X$ given $Y=y$ is $\E{X \mid Y}.$

The key tool used to reason about Lypaunov functions for our probabilisitc notion of ISS is a \textit{nonnegative supermartingale}, a specific type of expectation-governed random process: 

\begin{definition}
Let $\mb{x}_k$ be a sequence of random variables that take values in $\R^n$, $W:\mathcal{X}\times\mathbb{N}_{\geq 0}\to\R$, and suppose that $W(\mb{x}_k, k) \in L^1$
for $k\in \mathbb{N}_{\geq 0}$. The process $W_k\triangleq W(\mb{x}_k,k)$ is a supermartingale if:
\begin{equation}
    \label{eq:supermartingale}
        \mathbb{E}[ W_{k+1} \mid \mb{x}_{0:k}] \leq W_k~\textrm{almost~surely~for~all~}k\in\mathbb{N}_{\geq 0 },
\end{equation}
    where $\mb{x}_{0:k}$ indicates the random variables $\left\{\mb{x}_0, \mb{x}_1, \ldots, \mb{x}_k\right\}$. If, additionally, $W_k\geq 0$ for all $k\in\mathbb{N}_{\geq 0 } $, $W_k$ is a nonnegative supermartingale. If the process is non-decreasing in expectation, the process $W_k$ is a submartingale. If the inequality \eqref{eq:supermartingale} holds with equality, the process $W_k$ is a martingale. 
\end{definition}

An important result from martingale theory that we will use is \textit{Ville's inequality}, which bounds the probability that a nonnegative supermartingale rises above a certain value: 

\begin{theorem}[Ville's Inequality \cite{ville1939etude}]
    Let $W_k$ be a nonnegative supermartingale. Then for all $\lambda\in\R_{>0}$, 
    \begin{align}
        \mathbb{P} \left\{ \sup_{k\in \mathbb{N}_{\geq 0 }} W_k > \lambda  \right\} \leq \frac{\mathbb{E}[W_0]}{\lambda}.
        \label{eq:ville}
    \end{align}
\end{theorem}
Intuitively, Ville's inequality can be compared with Markov's inequality for nonnegative random variables; since the process $W_k$ is nonincreasing in expectation, Ville's inequality allows us to reason about the probability the process instead reaches some value above $\lambda$.



\cp

\section{Stability of Stochastic Discrete-Time Systems}
\label{sec:stability}

Traditional notions of stability may not necessarily apply to stochastic systems. For example, asymptotic stability to a point or forward invariance of a bounded set may be impossible in the presence of unbounded, stochastic disturbances. Thus more nuanced notions of stability are required \cite{kushner1974basic}. In this section we provide an abridged discussion of existing stability notions for stochastic systems. Notably, we discuss recurrence, boundedness of trajectories in probability and input-to-state stability for distributions with bounded support.

\subsection{Reccurence }

Recurrence is an important notion of stability used in the analysis of Markov chains \cite{meyn1993markov}. A bounded set $A \subset \mathcal{X}$ is recurrent if trajectories enter $A$ in finite time and visit $A$ infinitely often with probability 1 for all initial states $\mb{x}_0 \in \mathcal{X}$. 

\begin{definition}[Recurrence]
    For some bounded set $A \subset \mathcal{X}$ let the hitting time $\tau_A(\mb{x})  \triangleq \inf \{ k \in \mathbb{N}_{\geq 0} ~\textrm{s.t.}~ \mb{x}_k \in A, \; \mb{x}_0 = \mb{x} \} $. A set $A$ is recurrent if for every $\mb{x} \in \mathcal{X}$, $
        \mathbb{P}\{\tau_A(\mb{x}) < \infty  \} = 1. $ We say a system \eqref{eq:dt_dynamics} is recurrent if there exists a recurrent set $A$.
\end{definition}

\noindent Recurrence relates to the notion of stability for deterministic systems where trajectories remain within a set for all time, a property which is guaranteed for ISS systems. We refer the reader to \cite{meyn1993markov} for a more thorough treatment of Markov chain stability, recurrence, and ergodic theory.

\subsection{Boundedness in Probability}

Another notion of stability for stochastic systems is the probability that the state remains in a bounded region. Since it is often impossible to keep trajectories of \eqref{eq:dt_dynamics} bounded for all time \cite{steinhardt_finite-time_2012}, it is common to discuss these probabilities over some finite horizon $k \in \{ 0, \dots, K\} $ for some $K \in \mathbb{N}_{\geq 0 }$.   
\begin{definition}[Bounded in Probability]
    The system \eqref{eq:dt_dynamics} is bounded in probability for some $K \in \mathbb{N}_{\geq 0 }$ if there exists an $M>0$ and $ \epsilon \in (0,1)$ such that 
    \begin{align}
        \mathbb{P} \left\{ \max_{k \leq K} \Vert \mb{x}_k \Vert_{\ellp} \leq M \right\} \geq 1 - \epsilon. 
    \end{align}
\end{definition}
\noindent This notion of stability is central to Harold Kushner's work on on stochastic stability \cite{kushner_stochastic_1967} which we draw on for this paper, and which formed the basis for recent martingale-based approaches to finite-time stability \cite{steinhardt_finite-time_2012} and safety \cite{santoyo_verification_2019, santoyo_barrier_2021, cosner_robust_2023} for systems with unbounded uncertainty. This relates directly to the forward invariant region guaranteed to exist around the equilibrium point of ISS systems.

\subsection{ISS for Bounded Disturbance Distributions}

If the disturbance distribution $\mathcal{D}$ for system \eqref{eq:dt_dynamics} is only supported on a bounded set, then the essential supremum $\norm{\mathcal{D}}_{L^\infty}$ is well defined; thus if a system satisfies the  \ref{eq:ISS} condition \eqref{cor:iss2}, it is said to be stable in the ISS sense. Several authors have worked to extend ISS to the setting of unbounded stochastic disturbances. \cite{tsinias1998stochastic} proposed an ISS condition for continuous-time systems with unbounded disturbances, but required the disturbance magnitude to be upper bounded by a class-$\mathcal{K}$ of the state norm (thus, the disturbance vanishes at the equilibrium, a common but restrictive assumption). \cite{tang2020inputtostate, mcallister2021stochastic} also study stochastic variants of ISS, but only require that the \eqref{eq:ISS} condition hold for the expected trajectory (which does not guarantee boundedness of any trajectories).

\cp

\section{Input-to-State Stability for Unbounded Random Disturbances}

In this paper, we seek to generalize the notion of input-to-state stability to systems that are subject to unbounded random disturbances. Specifically, two issues arise when the support of $\mathcal{D}$ is unbounded: (i) the essential supremum $\norm{\mathcal{D}}_{L^\infty}$ may not be well defined, rendering the ISS condition inapplicable, and (ii) the probability that $\mb{x}_k$ remains in any bounded set for all $k \in \mathbb{N}$ is zero in general. 


This second point is somewhat non-intuitive; however, consider a system with additive Gaussian noise, $\mb{x}_{k+1} = \widehat{\mb{f}}(\mb{x}_k) + \mb{d}_k$, with $\mb{d}_k \sim \mathcal{N}(\mu, \boldsymbol{\Sigma})$ for some $\mu \in \R^n $ and $\boldsymbol{\Sigma} = \boldsymbol{\Sigma}^T > 0 $. Then, since the tails of $\mb{d}$ are unbounded, for any $B > 0,$ $\mathbb{P}\left\{\norm{\mb{d}}_\ellp > 2B \right\} = \epsilon  > 0$. This means, with probability $\epsilon$, $\norm{\mb{f}(\mb{x}, \mb{d})}_\ellp \geq \norm{\mb{d}}_\ellp - \norm{\widehat{\mb{f}}(\mb{x})}_\ellp > 2B - B = B$.\footnote{We must have $\norm{\widehat{\mb{f}}(\mb{x}_k)} \leq B$ for $\norm{\mb{x}_k} \leq B$; otherwise deterministic trajectories starting at $\mb{x}_k$ would leave the set in one step.} Thus, for any $K \in \mathbb{N}_{\geq 0}$, 
\begin{align*}
    \mathbb{P}\{\norm{\mb{x}_k}_\ellp < B, \; \forall k \leq K\} &\leq \mathbb{P}\{ \norm{\mb{d}_k}_\ellp \leq 2B, \; \forall k \leq K\}\\  &= (1-\epsilon)^K,
\end{align*} since all $\mb{d}_k$ are independent. Thus, as $K \to \infty$, the probability of the state remaining bounded goes to zero.

Thus, when generalizing ISS to the case of unbounded disturbances, we should expect weaker guarantees than those provided by the typical condition \eqref{eq:ISS}. With this in mind, we now define such a notion, \textit{Input-to-State Stability in Probability (ISSp)}, which is well-defined for systems subject to unbounded noise.  

\begin{definition}[Input-to-State Stable in Probability]
The system \eqref{eq:dt_dynamics} is \textit{input-to-state stable in probability (ISSp)} with repect to $L^p$ if, for any $\epsilon \in (0,1)$, $K \in \mathbb{N}_{\geq 0}$ and distribution $\mathcal{D} \in L^p$ such that $\Vert \mathcal{D}\Vert_{L^p}$, there exist functions $\beta \in \mathcal{KL}$, and $\gamma \in \mathcal{K}$ such that
\begin{align}
    \mathbb{P}\Big\{\norm{\mb{x}_k}_{\ellp} \leq \beta(\norm{\mb{x}_0}_{\ellp}, k) & + \gamma\big(\norm{\mathcal{D}}_{L^p}\big), \forall k \leq K \Big\} \nonumber \\
    &\geq 1 - \epsilon. \label{eq:issp}
\end{align}
If this holds for $\beta(\norm{\mb{x}_0},k) = M \alpha^k \norm{\mb{x}},$ for $M > 0, \alpha \in (0,1)$, the system is exponentially input-to-state stable in probability (ISSp).
\end{definition}

ISSp is a generalization of ISS to systems with (unbounded) stochastic disturbances. Intuitively, a system is ISSp if, for any disturbance in $L^p,$ and for any finite horizon $K$, there exist $\beta, \gamma$ such that the ISS condition \eqref{eq:ISS} (with the $L^\infty$ norm relaxed to the $L^p$) holds with a probability arbitrarily close to 1. 

As with ISS, we now relate ISSp to Lyapunov functions which can be used to verify this property.
\begin{definition}[ISSp Lyapunov Function]
    \label{def:issp_lyap}
    A continuous function $V: \R^n \to \R_\geq 0 $ is an \textit{ISSp Lyapunov Function} for the system \eqref{eq:dt_dynamics} if there exist functions $\kappa_1, \kappa_2, \kappa_3 \in \mathcal{K}_\infty$ and $\kappa_4 \in \mathcal{K}$ such that, 
    \begin{align}
        \kappa_1 (\Vert \mb{x} \Vert_{\ellp}) \leq V(\mb{x})  & \leq \kappa_2 (\Vert \mb{x}) \Vert_{\ellp}) \label{eq:v_bounds}\\
        \mathbb{E} [V(\mb{f}(\mb{x}, \mb{d}) - V(\mb{x}) ] & \leq - \kappa_3(V(\mb{x})) + \kappa_4(\Vert \mathcal{D} \Vert_{L^p})  \label{eq:delta_v}
    \end{align}
    for all $\mb{x} \in \mathcal{X}$ and $\Vert \mathcal{D}\Vert_{L^p} < \infty$. Additionally, if there exist constants $a, b, c > 0 $ and $\alpha \in (0,1)$ such that $\kappa_1(r) = a r^c , \kappa_2 (r) = b r^c ,$ and $ \kappa_3(r) = \alpha r$, then $V$ is an \textit{Exponential ISSp (E-ISSp) Lyapunov Function }  
\end{definition} 



\begin{remark}
As in the typical ISS definition \eqref{eq:ISS}, since $\max\{a, b \} \leq a + b \leq \max\{2a, 2b\}$, for suitable choices of $\beta, \gamma$, the ISSp condition \eqref{eq:issp} is equivalent to
\begin{align}
    \mathbb{P}\Big\{\norm{\mb{x}_k}_{\ellp} &\leq \max\left\{\beta(\norm{\mb{x}_0}_{\ellp}, k), \gamma\big(\Vert \mathcal{D} \Vert_{L^p}\big)\right\}, \; \forall k \leq K\Big\} \nonumber \\ &\geq 1 - \epsilon.
\end{align}

\end{remark}

In this paper, for simplicity of exposition, we will consider exponential ISSp.  Note that the results presented apply in the more general case, but the proofs become more complex.  


\cp 

\section{Lyapunov Conditions for E-ISSp}
As for ISS, there exist Lyapunov conditions for E-ISSp. To this end, we will use tools from martingale theory (in particular, Ville's inequality) to demonstrate that the existence of a Lyapunov function satisfying a drift condition in expectation implies a system is E-ISSp.\begin{theorem} 
    If there exists an E-ISSp Lyapunov function for system \eqref{eq:dt_dynamics}, then system \eqref{eq:dt_dynamics} is E-ISSp. \label{thm:lyap}
\end{theorem}

\begin{proof}
We begin by constructing a nonnegative supermartingale $W(\mb{x}_k, k)$ via a time-varying, affine transform of the Lyapunov function $V(\mb{x}_k)$. Rearranging the Lyapunov drift condition \eqref{eq:delta_v} with $\kappa_3(r) = \alpha r$ for some $\alpha \in (0,1) $ , we can see that $V(\mb{x}_k)$ almost resembles a supermartingale\footnote{Note that $\E{V(\mb{x}_{k+1}) \mid \mb{x}_k} = \E{V(\mb{x}_{k+1}) \mid \mb{x}_{0:k}}$ since system \eqref{eq:dt_dynamics} is Markovian.},
\begin{align}
    \E{V(\mb{x}_{k+1}) \mid \mb{x}_k} \leq (1 - \alpha) V(\mb{x}_k) + \varphi
\end{align}
where we define $\varphi \triangleq \kappa_4(\Vert \mathcal{D} \Vert_{L^p})\geq 0 $. However, this is not exactly a supermartingale due to the $(1-\alpha)$ scaling and the additive constant $\varphi$. 

Thus, for a particular horizon $K\in \mathbb{N}_{\geq 0 }$, we construct $W(\mb{x}_k, k)$ by undoing this scaling and translation. Letting $W_k \triangleq W(\mb{x}_k, k)$ for simplicity, this construction is:
\begin{align}
    W_k = \underbrace{\theta^k V(\mb{x}_k)}_{\textrm{rescale}} - \underbrace{\varphi \sum_{i=1}^k \theta^{i}}_{\textrm{translate}} + \underbrace{ \varphi \sum_{i=1}^K \theta^i}_{\textrm{ensure } W_k \geq 0 } , 
    \label{eq:w_def}
\end{align}
with $\theta \triangleq \frac{1}{1-\alpha} > 0 $ and the constant term $\varphi \sum_{i=1}^K \theta^{i}$ added to ensure $W_k \geq 0$.

Next we show $W_k$ is a nonnegative supermartingale.
We have $W_k \geq 0$ for any $\mb{x}_k \in \mathcal{X},$ since $V(\mb{x}_k) \geq 0$ by definition, and $\theta, \varphi \geq 0.$ Further, we have
\begin{align}
\mathbb{E}\big[&W_{k+1} \mid \mb{x}_k\big] = \mathbb{E} \left[\theta^{k+1} V(\mb{x}_{k+1}) + \varphi \sum_{i=k+2}^K \theta^i \right]\\
&\leq \theta^{k+1}\left((1-\alpha) V(\mb{x}_k) + \varphi\right) + \varphi \sum_{i=k+2}^K \theta^i\label{eq:mart_drift}\\
&= \theta^k V(\mb{x}_k) + \varphi \sum_{i=k+1}^K \theta^i = W_k,\label{eq:mart_simp}
\end{align}
where the inequality \eqref{eq:mart_drift} follows from the drift condition \eqref{eq:delta_v} and \eqref{eq:mart_simp} uses the fact that $\theta = \frac{1}{1-\alpha}.$

Since $W_k$ is a nonnegative supermartingale, we can apply Ville's inequality \eqref{eq:ville} to bound the probability $W_k$ that remains below any $\lambda > 0$ for all $k \leq K$.  Specifically, 
\begin{align}
    \mathbb{P}\Big\{ W(\mb{x}_k) \leq \lambda, \; \forall k \leq K \Big\} \geq 1 - \frac{W(\mb{x}_0)}{\lambda}.
    \label{eq:ville_w}
\end{align}
We also note that, using the geometric series identity $\sum_{i=1}^k \theta^{i-1} = \frac{\theta^k-1}{\theta -1}$, we can write $W_k$ as
\begin{align}
    W_k = \theta^k V(\mb{x}_k) + \frac{\theta \varphi (\theta^K - \theta^k)}{\theta - 1} \label{eq:w_geom}
\end{align}

Examining the structure of $W_k,$ if for all $k \leq K$ we have $W_k \leq \lambda,$ rearranging this inequality results in
\begin{align}
    V(\mb{x}_k) &\leq \left(\lambda - \frac{\theta^{K+1} \varphi}{\theta - 1}\right) \theta^{-k} + \frac{\theta }{\theta - 1}  \varphi\\
    & \leq (\lambda - \varphi ) \theta^{-k} + \frac{\theta}{\theta -1} \varphi \label{eq:theta_bound1}\\
    &\triangleq \left( M \norm{\mb{x}_0}_{\ellp}^c + \dtwo \varphi \right)\theta^{-k}  +  \frac{\theta}{\theta - 1} \varphi \label{eq:lambda_def}\\
    & \leq M \norm{\mb{x}_0}_{\ellp}^c \theta^{-k} + \dtwo \varphi  + \frac{\theta}{\theta - 1} \varphi \label{eq:theta_bound2}
\end{align}
Inequality \eqref{eq:theta_bound1} follows from $\theta > 1$ and $\varphi \geq 0  $. Equality \eqref{eq:lambda_def} follows from choosing $\lambda = M \norm{\mb{x}_0}^c + \left( 1 + \dtwo \right) \varphi$  
for some $ \dtwo \geq 0 $ and $M, \;c> 0$ . Inequality \eqref{eq:theta_bound2} is due to $\theta >1$ and $\varphi, \eta \geq 0 $. 

Further, using the lower bound \eqref{eq:v_bounds} on $V(\mb{x}_k)$ and the definition of $\theta$, \eqref{eq:theta_bound2} becomes
\begin{align}
    a \norm{\mb{x}_k}_{\ellp}^c \leq M \norm{\mb{x}_0}_{\ellp}^c (1-\alpha)^k + \left(\dtwo + \frac{1}{\alpha} \right)  \varphi
\end{align}
for some $a >0$ which, rearranging, and raising both sides to the power of $\frac{1}{c}$ (which preserves order since $c >0$), yields
\begin{align}
    \norm{\mb{x}_k}_{\ellp} &\leq \left(\frac{M}{a} \norm{\mb{x}_0}_{\ellp}^c (1-\alpha)^k + \frac{(\dtwo + \frac{1}{\alpha})}{a} \varphi \right)^{\frac{1}{c}}\\
    &\leq \left(\frac{M}{a}\right)^\frac{1}{c} \zeta \norm{\mb{x}_0}_{\ellp} (1 - \alpha)^{\frac{k}{c}} + \zeta \left(\frac{(\dtwo + \frac{1}{\alpha})\varphi}{a}\right)^{\frac{1}{c}}  \label{eq:power_split}\\
    &\triangleq \tilde{M} \tilde{\alpha}^k \norm{\mb{x}_0}_{\ellp} + \gamma_\eta (\Vert \mathcal{D}\Vert_{L^p}) ,
\end{align}
for $\tilde{M} \triangleq \zeta \left( M / a \right)^\frac{1}{c} > 0$, $\tilde{\alpha} \triangleq (1 - \alpha)^\frac{1}{c}\in (0, 1) $, some $\zeta>0$ as needed, and $\gamma_\eta(r) \triangleq \zeta\left( \frac{\dtwo + \frac{1}{\alpha}}{a}\right)^\frac{1}{c}\kappa_4(r)^\frac{1}{c}$ which is a class $\mathcal{K}$ for all $\dtwo \geq 0$. The existence of $\zeta$ follows from Lemma \ref{lemma:power_split} shown in the Appendix. 

Thus, we now must ensure there exists a suitable choice of $M, \eta$ such that the probability that this bound holds for all $k \in \{ 1, \ldots, K \}$ is greater than $1-\epsilon.$ By Ville's inequality, 
\begin{align}
    \mathbb{P}&\left\{\norm{\mb{x}_k}_{\ellp} \leq \tilde{M} \tilde{\alpha}^k \norm{\mb{x}_0}_{\ellp} +\gamma_\dtwo(\Vert \mathcal{D} \Vert_{L^p}) , \; \forall k \leq K \right\} \nonumber \\ & \geq 1 - \frac{W_0}{\lambda} = 1 - \frac{V(\mb{x}_0) + \frac{\varphi}{\alpha} ((1-\alpha)^{-K} - 1)}{M \norm{\mb{x}_0}_{\ellp}^c + (1 + \eta) \varphi}, \label{eq:prob_bound}
\end{align}
with $W_0 \geq 0$ by definition. Thus, as long as $\norm{\mb{x}_0}_{\ellp}^c,$ $\norm{\mathcal{D}}_{L^p}$ are not both zero, we can choose $M, \eta$ large enough to have
\begin{align*}
    \mathbb{P}\Big\{ \norm{\mb{x}_k}_{\ellp} \leq  \tilde{M} \tilde{\alpha}^k \norm{\mb{x}_0}_{\ellp} + & \gamma_\dtwo(\Vert  \mathcal{D} \Vert_{L^p}), \; k \leq K \Big\}  \geq 1 - \epsilon, \label{eq:loose_bound}
\end{align*}
for any $\epsilon \in (0,1)$, so the system must be E-ISSp. 

\end{proof}

\begin{remark}
    The variables $M\geq 0 $ and $\dtwo \geq 0 $ are free parameters which can be varied to analyze the probability of convergence and boundedness respectively. We note that the bound in \eqref{eq:loose_bound} may be very weak; stronger bounds can be achieved by removing the bounding steps in \eqref{eq:theta_bound1} and \eqref{eq:theta_bound2}, but clarity was chosen over tightness for this proof. 
\end{remark}

\cp 

\section{Connections to other Stability Notions}
Here we discuss connections between ISS, ISSp, and other notions of stability for stochastic systems, as surveyed in Section \ref{sec:stability}. 

\begin{corollary}
\label{cor:iss-infty}
If the system \eqref{eq:dt_dynamics} is ISS, then it is ISSp with respect to $L^\infty$. 
\end{corollary}
\begin{proof}
    By definition, if a system is ISS, then for all $k \in \mathbb{N}_{\geq 0},$ there exist $\beta \in \mathcal{KL}, \gamma \in \mathcal{K}$ such that 
    \begin{align}
        \norm{\mb{x}_k}_{\ellp} \leq \beta(\norm{\mb{x}_0}_{\ellp}, k)  + \gamma \left( 
        \vartheta
        \right).
    \end{align}
    for all $\vartheta \geq \sup_{k \in \mathbb{N}_{\geq 0}} \Vert \mb{d}_k \Vert$. 
    
    Thus, since the $L^\infty$-norm (equivalently, the essential supremum) is finite for all random variables in $L^\infty,$ for any $\mb{d} \sim \mathcal{D}$ with $\mb{d} \in L^\infty,$ we have:
    \begin{align*}
        \mathbb{P}\Big\{\norm{\mb{x}_k}_{\ellp} \leq \beta(\norm{\mb{x}_0}_{\ellp}, k) + \gamma(\norm{\mathcal{D}}_{L^\infty}), \; \forall k \in \mathbb{N} \Big\} = 1. 
    \end{align*}
    Thus trivially, for any $K \in \mathbb{N}_{\geq 0 }$, $\epsilon \in (0, 1)$, we have $\beta \in \mathcal{KL}, \gamma \in \mathcal{K}$ such that \eqref{eq:issp} holds for all distributions $\mathcal{D}$ with $\mb{d} \in L^\infty$ . Thus, the system is ISSp w.r.t. $L^\infty$.
 \end{proof}
Corollary \ref{cor:iss-infty} provides a clear connection between ISS and ISSp: if a system is ISS, it by definition is ISSp for disturbances in $L^\infty$.
Next, we discuss the relationship between ISS and ISSp w.r.t. $L^2$ which is a much larger class of unbounded random variables.

\begin{corollary}
    If the system \eqref{eq:dt_dynamics} is additive with respect to its disturbance and admits a twice-continuously differentiable, convex E-ISS Lyapunov function $V: \R^n \to \R_{\geq 0}$ such that $\sup_{\mb{x} \in \mathcal{X}}  \Vert \nabla^2 V(\mb{x}) \Vert_2  \leq  \lambda_{\textrm{max}}$ for some $\lambda_{\textrm{max}} \geq 0 $ then it is E-ISSp for $\mb{d} \in L_p$ with $\E{\mb{d}} = 0 $ for  $p \geq 2$. 
    \label{cor:iss2}
\end{corollary}
\begin{proof}
    The dynamics are additive with respect to the disturbance so system \eqref{eq:dt_dynamics} can be rewritten as: 
    \begin{align}
        \mb{x}_k =  \mb{f}(\mb{x}_k, \mb{d})  \triangleq \widehat{\mb{f}}(\mb{x}_k)  + \mb{d} \label{eq:d_add_dyn}
    \end{align}
    The function $V$ is a E-ISS Lyapunov function for \eqref{eq:d_add_dyn} so it satisfies: 
    \begin{align}
        V\left(\widehat{\mb{f}}(\mb{x})  + \mb{d} \right) - V(\mb{x}) \leq -\alpha V(\mb{x}) + \kappa_4(\Vert \mathcal{D} \Vert_{L^\infty}) \label{eq:pf_iss}
    \end{align}
     for all $\mb{x} \in \mathcal{X}$, some $\alpha \in (0, 1) $, $\sigma \in \mathcal{K}$, and any $\mb{d} \in L^\infty$. 

    The expected value of the left side of this inequality is: 
    \begin{align}
        & \E{V\left(\widehat{\mb{f}}(\mb{x})  + \mb{d} \right) - V(\mb{x})}  = \E{\mb{V}\left(\widehat{\mb{f}}(\mb{x})  + \mb{d} \right) } - V(\mb{x}) \nonumber \\ 
        & \leq V\left(\widehat{\mb{f}}(\mb{x})  + \E{\mb{d}} \right) - V(\mb{x}) + \frac{\lambda_{\textrm{max}}}{2}\textrm{tr}(\textrm{cov}(\mb{d})) \label{eq:jensen_gap}\\
        & = V\left(\widehat{\mb{f}}(\mb{x}) + 0\right) - V(\mb{x}) + \frac{\lambda_{\textrm{max}}}{2}\textrm{tr}(\textrm{cov}(\mb{d})) \label{eq:d_mean_zero}\\
        & \leq - \alpha V(\mb{x})  + \frac{\lambda_{\textrm{max}}}{2} \textrm{tr}(\textrm{cov}(\mb{d})) \label{eq:apply_ISS_for_Ed}
    \end{align}
    where \eqref{eq:jensen_gap} accounts for Jensen's inequality as in \cite[Lemma 1]{cosner_robust_2023}, \eqref{eq:d_mean_zero} is due to the assumption that the $\E{\mb{d}} = 0$, and \eqref{eq:apply_ISS_for_Ed} is an application of the E-ISS bound \eqref{eq:pf_iss}.  

    Since bounded covariance implies boundedness in $L^2$, if $\mb{d} \in L^2 $ then $V$ is an E-ISSp Lyapunov function for \eqref{eq:d_add_dyn}. Furthermore, since $L^2 \supseteq L_p$ for  $p \geq 2$, $V$ is an E-ISSp Lyapunov function for $\mb{d} \in L^p$ for all $p \geq 2$. 
\end{proof}

Next we discuss the relationship between ISSp and trajectories that are bounded in probability. 
 

\begin{corollary}
    If system \eqref{eq:dt_dynamics} is ISSp  w.r.t. $L^p$, then for any $\mb{d} \in L^p,$ the system's trajectories are bounded in probability.   
\end{corollary}
\begin{proof}
    If the system is ISSp w.r.t. $L^p$, then for any $K \in \mathbb{N}_{\geq 0}$ and $ \epsilon \in (0, 1)$, there exist $\beta \in \mathcal{KL}$ and $ \gamma \in \mathcal{K}$ such that
    \begin{align*}
        \mathbb{P}\bigg\{\norm{\mb{x}_k}_{\ellp} \leq \beta(\norm{\mb{x}_0}_{\ellp}, k) + \gamma\left(\Vert\mathcal{D} \Vert_{L^p}\right), \; \forall k \leq K \bigg\} \geq 1 - \epsilon.
    \end{align*}
    Then, since $\beta$ is decreasing in $k$, we have that for $B_0 \triangleq  \beta(\norm{\mb{x}_0}_{\ellp}, 0) + \gamma\left(\Vert \mathcal{D} \Vert_{L^p}\right),$
    \begin{align}
        \mathbb{P}\left\{\norm{\mb{x}_k}_{\ellp} \leq B_0, \; \forall k \leq K \right\} \geq 1 - \epsilon.
    \end{align}
    Thus the system trajectories are bounded in probability.
\end{proof}

\noindent Like with traditional ISS, the ISSp condition \ref{eq:issp} is equivalent to system trajectories remaining in a ball whose radius scales with the initial condition and the norm of the disturbance. Thus, if a system is ISSp, its trajectories (over a finite horizon) must be bounded in probability. 

Finally, we look to discuss the relationship between ISSp and recurrence. To do this we, first restate an important result from drift analysis (see \cite{lengler2020drift} for a detailed survey).

\begin{theorem}[Variable Drift \cite{lengler2020drift}]
\label{thm:variable_drift}
Suppose there exists some function $V: \mathcal{X} \to \mathbb{R}_{\geq 0}$, with $\gamma$-sublevel set $V_\gamma \triangleq \left\{ \mb{x} \in \mathcal{X} \mid V(\mb{x}) \leq \gamma \right\}$ such that for all $\mb{x} \in \mathcal{X} \setminus V_\gamma,$
\begin{align}
    \E{V\left(\mb{f}(\mb{x}, \mb{d})\right) - V(\mb{x})} \leq -h(V(\mb{x})), \label{eq:variable_drift}
\end{align}
for some increasing function $h: \mathbb{R}_{>0} \to \mathbb{R}_{>0}$. Then, for any trajectory with initial state $\mb{x}_0,$ the hitting time $\tau_\gamma(\mb{x}_0) = \inf\{k \mid V(\mb{x}_k) \leq \gamma \}$ is bounded in expectation by
\begin{align}
    \E{\tau_\gamma(\mb{x}_0)} \leq \frac{\gamma}{h(\gamma)} + \int_{\gamma}^{V(\mb{x}_0)} \frac{1}{h(\sigma)}d\sigma.
\end{align}
\end{theorem}
Using this result, we can show that, if a system admits an E-ISSp Lyapunov function, then any Lyapunov sublevel set (above a particular value) must be recurrent. 
\begin{theorem}
If there exists an E-ISSp Lyapunov function w.r.t. $L^p$ for system \eqref{eq:dt_dynamics}, then \eqref{eq:dt_dynamics} is recurrent.
\end{theorem}
\begin{proof}
    Suppose there exists an E-ISSp Lyapunov function $V$ for the system \eqref{eq:dt_dynamics}. Then, for $h(V(\mb{x}_k)) = \alpha V(\mb{x}_k) - \varphi,$ with $\varphi \triangleq \kappa_4(\Vert \mathcal{D} \Vert_{L^p})$, we have  
    \begin{align}
        \E{V(\mb{x}_{k+1}) - V(\mb{x}_k) \mid \mb{x}_k} \leq -h(V(\mb{x}_k)).
    \end{align}
    For any $\gamma > \frac{\varphi}{\alpha}$,  $h(V(\mb{x})) > 0$ for all $\mb{x} \in \mathcal{X} \setminus V_\gamma;$ thus our system meets the variable drift condition \eqref{eq:variable_drift}. 

    Thus, consider some trajectory with an initial state $\mb{x}_0 \in \mathcal{X} \setminus V_\gamma.$ Then, by Theorem \ref{thm:variable_drift}, we have 
    \begin{align}
        \E{\tau_\gamma(\mb{x}_0)} &\leq \frac{\gamma}{\alpha \gamma - \varphi} + \int_{\gamma}^{V(\mb{x}_0)} \frac{1}{\alpha \sigma - \varphi} d\sigma\\
        &\leq \frac{\gamma}{\alpha \gamma - \varphi} + \frac{1}{\alpha}\log\left(\frac{\alpha V(\mb{x}_0) - \varphi}{\alpha \gamma - \varphi}\right) < \infty.
    \end{align}
    Since $\E{\tau_\gamma(\mb{x}_0)} < \infty, $ we must have $\mathbb{P}\{\tau_\gamma(\mb{x}_0) < \infty\} = 1.$ Thus, for any $\gamma > \frac{\alpha}{\varphi},$ the sublevel set $V_\gamma$ is recurrent. Since $\kappa_1$ is radially unbounded, $V_\gamma$ must be bounded for all $\gamma\geq 0 $, thus the system is recurrent. 
\end{proof}






\cp 
\section{Practical Example: Linear-Quadratic-Gaussian Control}
\begin{figure}[t]
     \centering
     \begin{subfigure}[b]{0.375\textwidth}
         \centering
         \includegraphics[width=\textwidth]{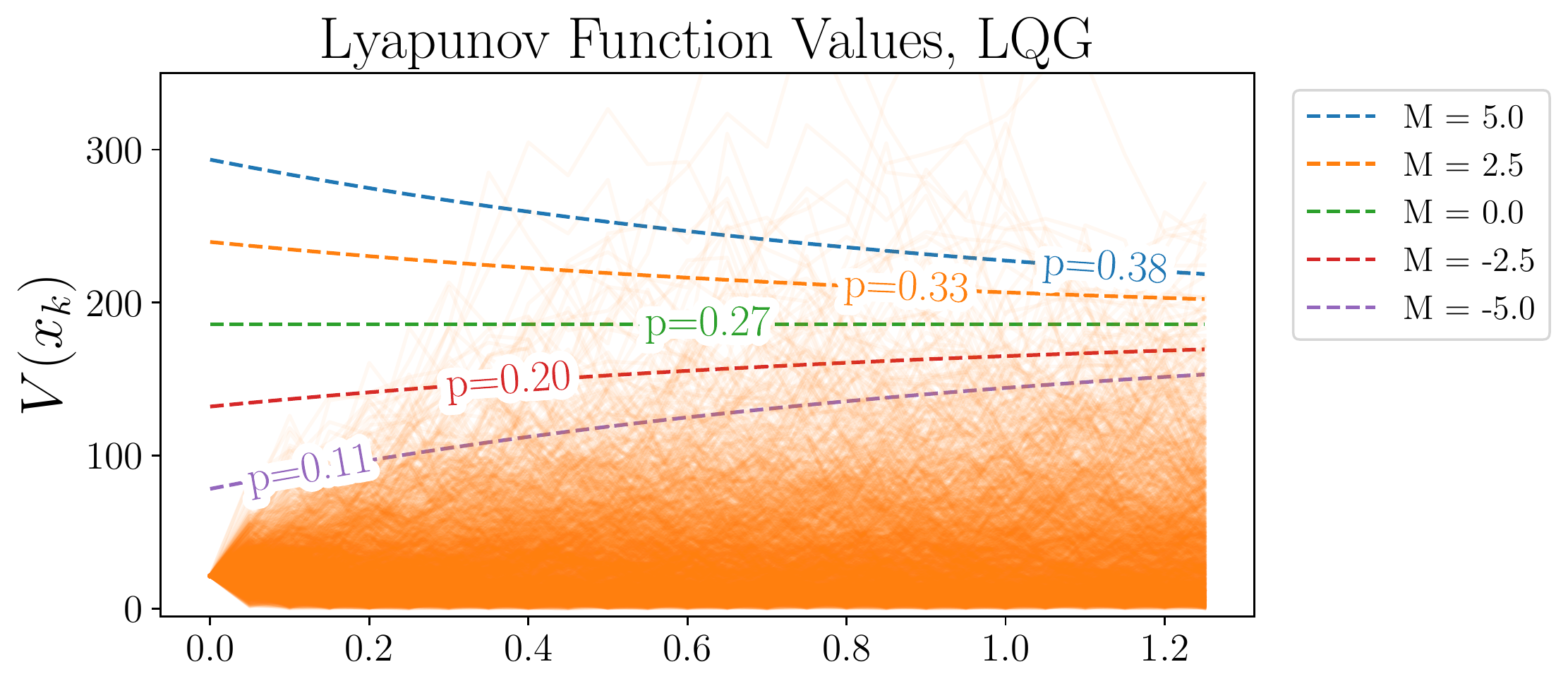}
     \end{subfigure}
     \\
     \begin{subfigure}[b]{0.375\textwidth}
         \centering
         \includegraphics[width=\textwidth]{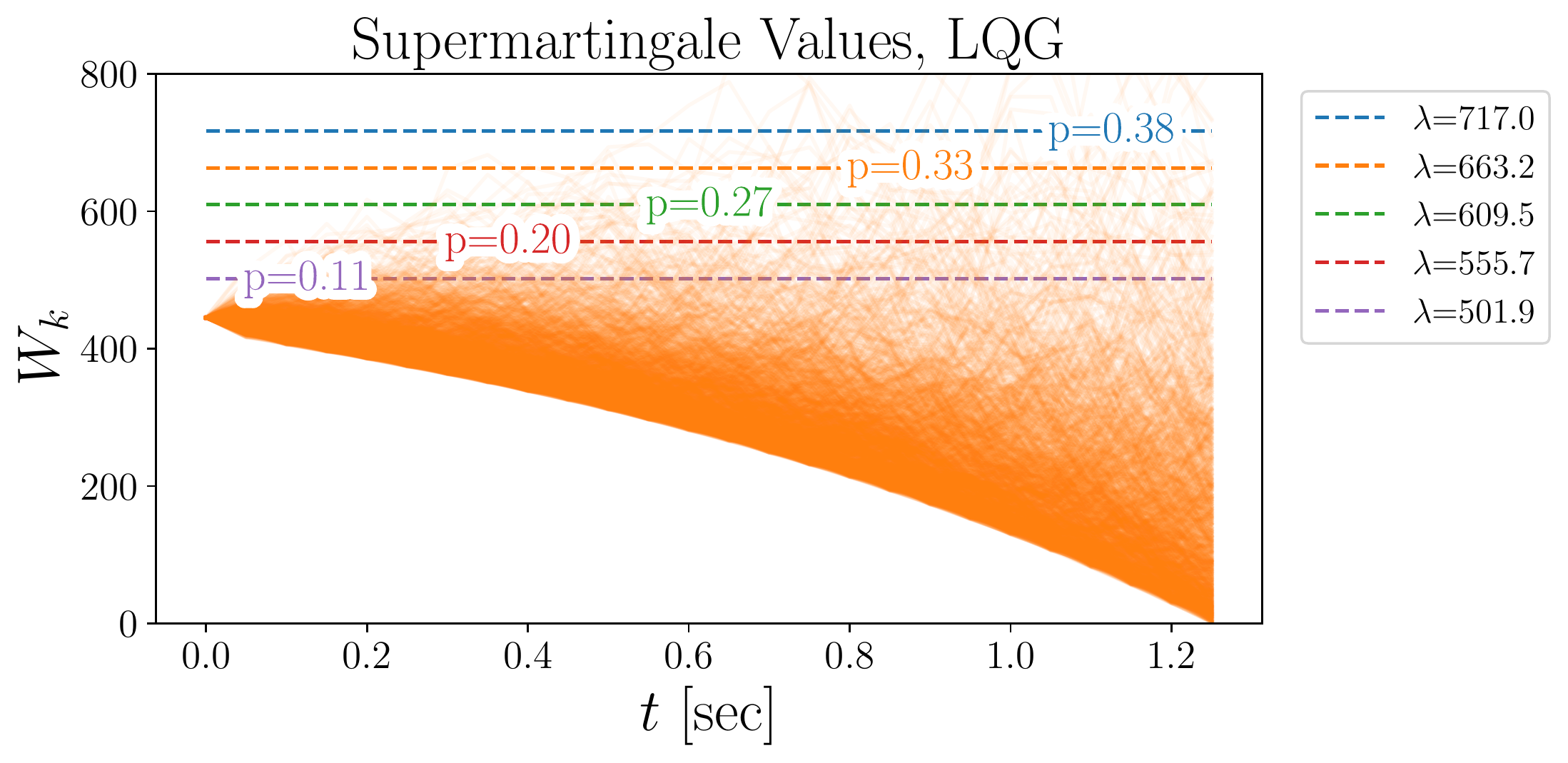}
     \end{subfigure}
     \\
    \begin{subfigure}[b]{0.18\textwidth}
         \centering
         \includegraphics[width=\textwidth]{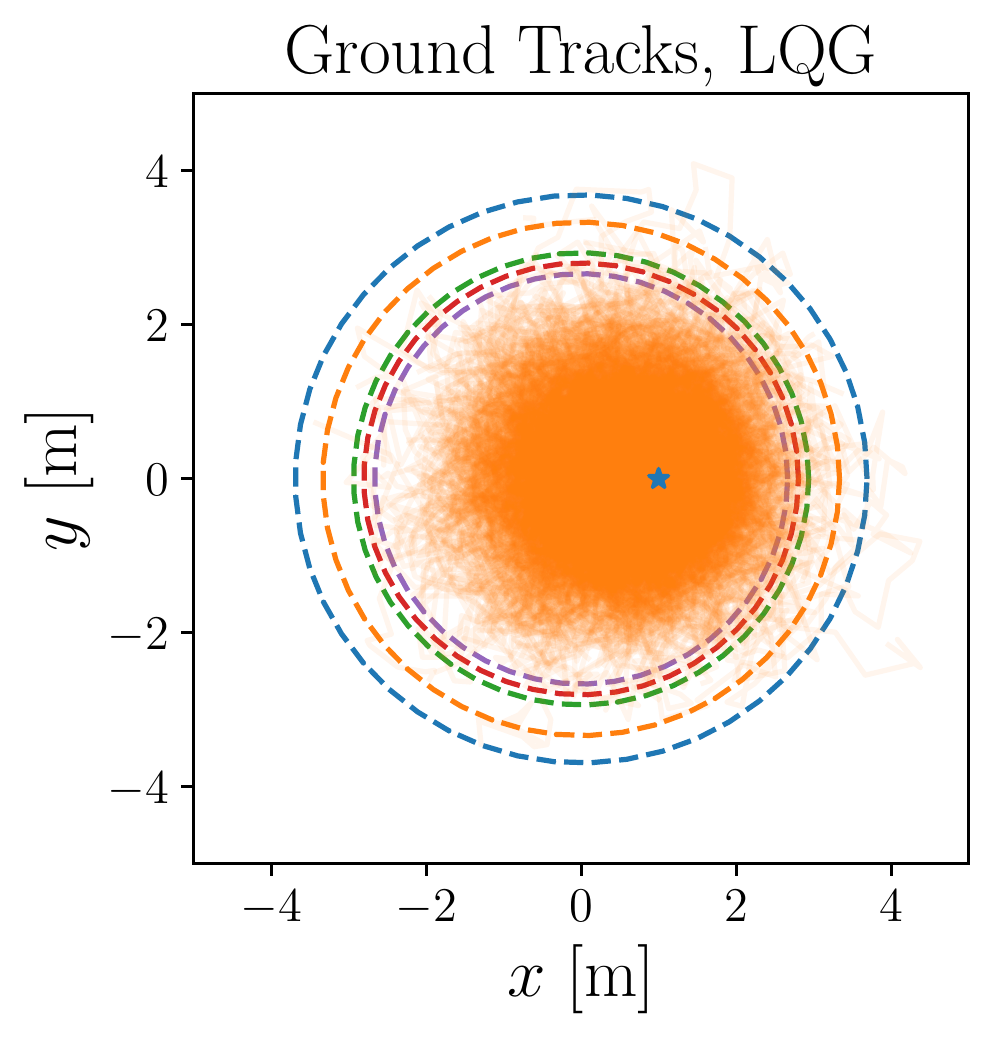}
     \end{subfigure}
     \;
     \begin{subfigure}[b]{0.18\textwidth}
         \centering
         \includegraphics[width=\textwidth]{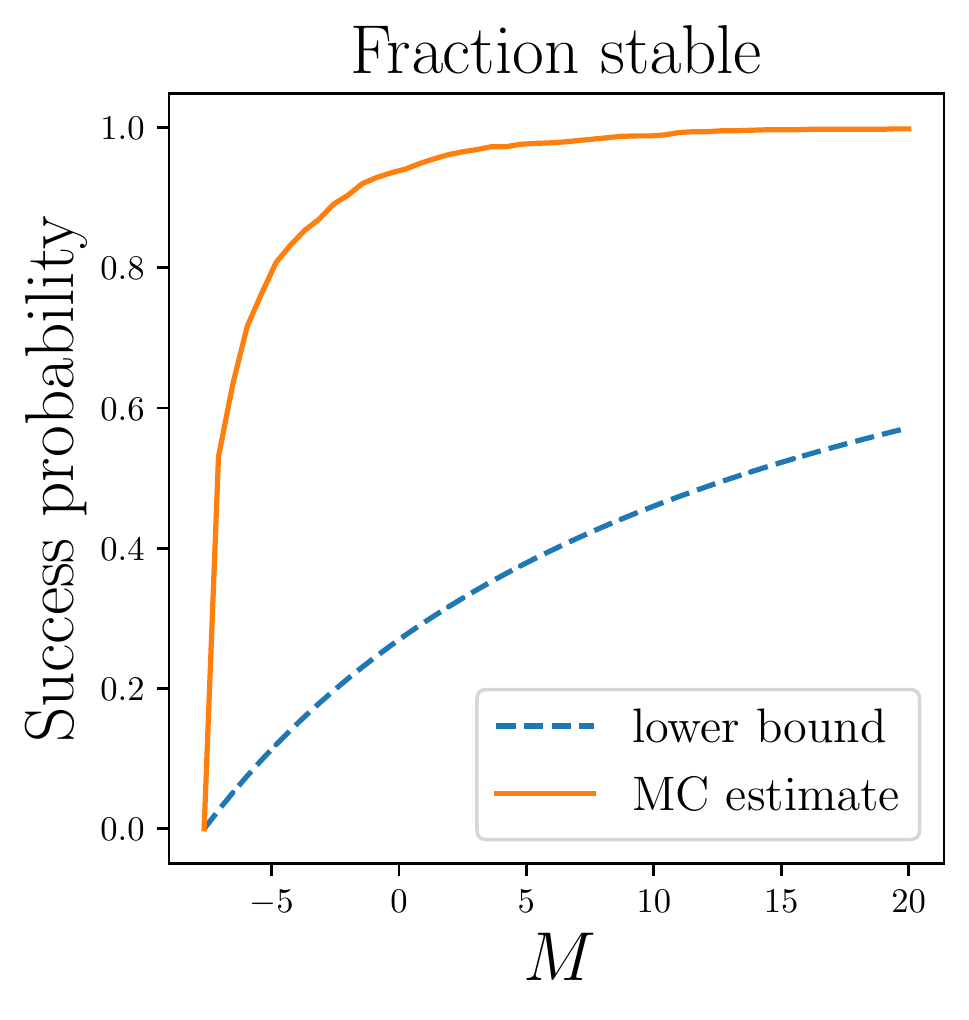}
     \end{subfigure}
     \caption{Simulation results for the double integrator over 1500 trials. \textbf{(Top):} Lyapunov function values $V(\mb{x}_k)$ plotted against the upper bound trajectory for various $M$. \textbf{(Middle):} Corresponding supermartingale values $W_k$ for each trajectory. Note the geometric upper bound trajectories for $V(\mb{x})$ correspond one-to-one to the level sets of $W_k$. \textbf{(Bottom left):} Ground tracks for the double integrator position. \textbf{(Bottom right):} Monte Carlo estimation of the ``success fraction'' (fraction of trajectories with $V(\mb{x}_k) \leq \rho_k$) versus the guaranteed lower bound.}
     \label{fig:lqg}
     \vspace{-2em}
\end{figure}
\label{sec:LQG}
We now consider the ISSp properties of some practical systems, and study the validity of our exit probability bounds via simulation. We begin by considering the case of linear time-invariant systems subject to additive, zero-mean Gaussian noise. The system dynamics are given by $\mb{x}_{k+1} = \mb{A} \mb{x}_k + \mb{B} \mb{u}_k + \mb{d}_k,$ where $\mb{u}_k \in \mathbb{R}^m$ is a control input to the system, and $\mb{d}_k \overset{\text{i.i.d.}}{\sim} \mathcal{N}(\mb{0}, \boldsymbol{\Sigma}) = \mathcal{D}$ for some $\boldsymbol{\Sigma} = \boldsymbol{\Sigma}^T > 0 $. Suppose our system implements an infinite-horizon LQR feedback policy, i.e., $\mb{u}_k = - \mb{K} \mb{x}_k,$ where $\mb{K} = (\mb{R} + \mb{B}^T \mb{P} \mb{B})^{-1} (\mb{B}^T \mb{P} \mb{A})$ for $\mb{P}$ satisfying the discrete-time algebraic Riccati equation.
For this closed-loop system, $V(\mb{x}) = \mb{x}^T \mb{P} \mb{x}$ is an E-ISS Lyapunov function, with 
\begin{align}
    \E{V(\mb{x}_{k+1}) - V(\mb{x}_k)} \leq -\alpha V(\mb{x}_k) + \sigma \norm{\mathcal{D}}_{L^2}^2
\end{align}
for $\alpha = \frac{\lambda_\text{min}(\mb{Q})}{\lambda_\text{max}(\mb{P})},$ $\sigma = \lambda_\text{max}(\mb{P})$, and any $ \mathcal{D} \in L^2$. In particular, if we pick $\lambda = M V(\mb{x}_0) + \frac{\sigma \norm{\mathcal{D}}_{L^2}^2}{\alpha (1 - \alpha)^K},$ we can bound the probability $V_k$ rises above a time-varying trajectory:
\begin{align}
    \mathbb{P}\bigg\{V(\mb{x}_k) &\leq \rho_k, \; \forall k \leq K \bigg\} \geq 1 - \frac{W_0}{\lambda},
    \label{eq:ville_lqg}
\end{align}
with $\rho_k = M V(\mb{x}_0) (1-\alpha)^k + \frac{\sigma\norm{\mathcal{D}}_{L^2}^2}{\alpha}$, and $W_0$ defined as in \eqref{eq:w_def}. Specifically, we study how this bound varies numerically for a double integrator system in the plane (see \cite{cosner_robust_2023} for a detailed dynamics derivation). Figure \ref{fig:lqg} plots the results of 1500 simulations of the double integrator. First, we plot the values of $V(\mb{x}_k)$ and $W_k$ across multiple choices of $M.$ We can see that the trajectories $W_k = \lambda$ correspond exactly to $V(\mb{x}_k) = \rho_k$; thus, intuitively, the event $W_k > \lambda$ is exactly the event where $V_k$ rises above a shifted geometric sequence $\rho_k$. 

We also plot the trajectories of the system in the plane, along with Lyapunov level sets $V_\rho$ for $\rho = \max_k \rho_k$, evaluated when the system's velocity is zero. We choose the maximum value of $\rho_k$ since if $V_k \leq \rho_k$ for all $k \leq K$, we must have $V_k \leq \max_k \rho_k$. Interestingly, one can show the probability bound \eqref{eq:ville_lqg} is equivalent to the exit probability bound provided by Kushner \cite{kushner_stochastic_1967}, up to a choice of scaling to construct $W_k$ (we explore this point further in Section \ref{sec:walking}).

Finally, we plot our bound \eqref{eq:ville_lqg} on the ``success probability'' $\mathbb{P}\{ V_k \leq \rho_k, \; \forall k \leq K\}$ versus the fraction of trajectories that remained under $\rho_k$ for various choices of $M$. While we note that our bound is, indeed, a lower bound on the success probability, it is quite a weak lower bound; thus, finding Lyapunov functions and martingales that yield stronger bounds is an interesting direction for future work. 

\section{Practical Example: Seven-Link Walker}
\label{sec:walking}
Consider the seven-link walker as shown in Fig. \ref{fig:walker-config}. As detailed in \cite{tucker2023input}, walking can be distilled down to the discrete-time dynamical system described by the Poincar\'e return map:
\begin{align}
    \mathscr{P} : & B_{\rho}(\mb{x}^*) \times [d^-,d^+] \partialto S_{ [d^-,d^+]} : = \bigcup_{d \in [d^-,d^+]} S_d, \nonumber\\
    & \mb{x}_{k+1} = \mathscr{P}(\mb{x}_k,d_k), \qquad   d_k \in  [d^-,d^+],
    \label{eq: discretepoincare}
\end{align}
for some sequence of step heights $d_k \in [d^-,d^+] \subset \mathbb{R}$, $k \in \mathbb{N}_{\geq 0}$ and $S_d \subset \R^n$ denoting the uncertain guard condition:
\begin{align}
    S_d &= \{\mb{x}\in \R^n \mid h(\mb{x}) = d, ~\dot h(\mb{x}) < 0\},
    \label{eq: heightguard}
\end{align}
where $h: \R^n \to \R$ is typically selected to denote the vertical height of the swing foot relative to the stance foot. Note that the partial function nature of $\mathscr{P}$ implies that solutions may not exist for all time, i.e., the solution $\mb{x}_k$ might leave the ball $B_{\rho}(\mb{x}^*)$ on which $\mathscr{P}$ is well-defined.

\begin{figure}[tb]
    \centering
    \includegraphics[width=\linewidth]{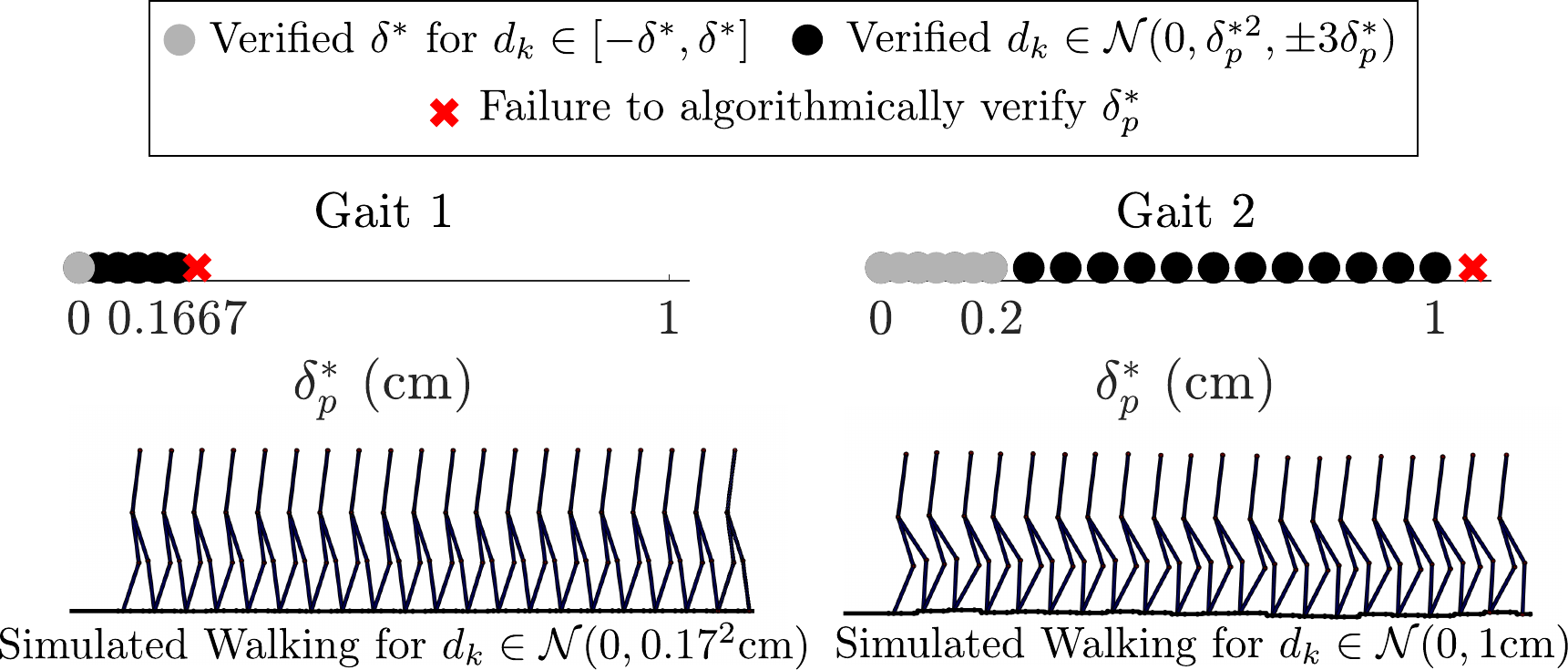}
    \caption{Algorithmic results of ISSp compared to ISS. As shown, the stochastic ISSp condition yields more realistic predictions of the tolerable step heights for two gaits (the same gaits as those compared in \cite{tucker2023input}).}
    \label{fig: walkingalg}
     \vspace{-2em}
\end{figure}

\newsec{Guaranteed Robustness to Uncertain Terrain} Prior work considered input-to-state stability of \eqref{eq: discretepoincare} with bounded step heights described as the set $\mathbb{D} \triangleq [-\delta,\delta] \subset \mathbb{R}$ with $\delta > 0$ \cite{tucker2023input}. Moreover, leveraging this discrete-time representation for bounded disturbances \cite{tucker2023input} introduced an ISS perspective on bipedal locomotion. Explicitly, a periodic walking gait with a nominal fixed point $\mb{x}^* = \mathscr{P}(\mb{x}^*,0)$ is defined as $\delta$-robust
for a given $\delta>0$ if for the discrete-time dynamical system \eqref{eq: discretepoincare}, with any $d_k \in [-\delta,\delta]$, there exists some forward invariant set $\mathcal{W} \subset B_{\rho}(\mb{x}^*)$ such that for all $\mb{x}_0 \in \mathcal{W}$, the system is ISS. Moreover, this definition of robustness was shown to be verifiable through an ISS Lyapunov function. 


Specifically, Theorem 2 of \cite{tucker2023input} states that if the Lyapunov condition as in \eqref{eq:iss_lyap} is satisfied, then the periodic gait is ISS.
To verify this, a candidate Lyapunov function can be synthesized by approximating the exponentially stable discrete-time system using the linearization of the Poincar\'e return map for $d_k = 0$:
$$
\mb{x}_{k + 1} = \mb{A} \mb{x}_k  := D\mathscr{P}(0,0) \mb{x}_k.
$$
Then, the Lyapunov matrix $\mb{P} = \mb{P}^T > 0$ is obtained by solving the discrete-time Lyapuov equation 
($\mb{A}^T \mb{P} \mb{A} - \mb{P} = - \mb{Q}$) 
for $\mb{Q} = \mb{Q}^T > 0$ which provides a discrete-time Lyapunov function $V(\mb{x}) = \mb{x}^T \mb{P} \mb{x}$.

\newsec{Probabilistic Robustness to Uncertain Terrain} To obtain more reasonable estimates of the maximum step heights that a given periodic gait can withstand, in this work we will instead consider step heights drawn from some distribution $d_k \sim \mathcal{D}$ and apply the ISSp methodology. Specifically, we take $\mathcal{D} := \mathcal{N}(0,\delta_p^2)$ such that $\delta_p > 0$ now represents the standard deviation of the distribution. Note that due to the partial nature of the $\mathscr{P}$, we will truncate $\mathcal{D}$ at $3\delta_p$ to ensure that there exists some $\delta_p$ such that $\mathbb{E} [\Delta V] < + \infty$, we will denote this truncated Gaussian as $\mathcal{N}(\cdot,\cdot,\pm a)$ where $a$ denotes the truncation interval \cite{grimmett_probability_2020}.

\begin{figure}[tb]
    \centering
    \includegraphics[width=0.99\linewidth]{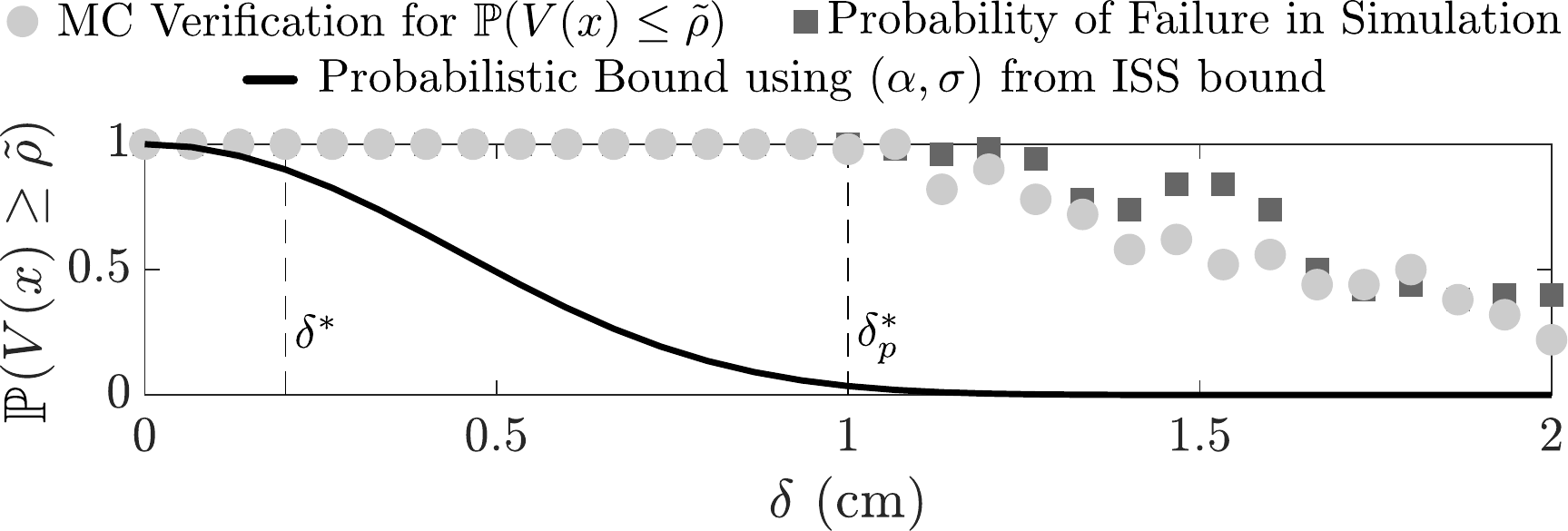}
    \caption{The probabilistic bound, evaluated for Gait 2, quickly decreases as $\delta$ increases. However, the simulation results show that the walking is able to remain periodic for all values determined to be $\delta$-robust in Opt. \ref{eq: optdelta}, highlighting that ISSp yields more reasonable estimates of $\delta$ compared to the strict ISS condition presented in \cite{tucker2023input}.}
    \label{fig: probbound2}
     \vspace{-4mm}
\end{figure}

Consider the ISSp Lyapunov condition from Def. \eqref{def:issp_lyap}:
\begin{align}
    \mathbb{E} [V(\mathscr{P}(\mb{x},d)) -V(\mb{x}) \mid \mb{x}]  \leq -\tilde{\alpha} V(\mb{x}) + \tilde{\sigma},
    \label{eq: ExpLyapWalking2}
\end{align}
with $\tilde{\alpha} = 2k\lambda_{\max}(\mb{P})$, and $\tilde{\sigma} = k(\chi\delta_p)^c$, where $k \in (0,1)$ is a user-defined variable dictating the convergence of the Lyapunov condition and $\chi > 0$ is used to scale the set over which the expected exponential decay condition holds. Note that this condition can be equivalently expressed in the form:
\begin{align}
\|\mb{x} &- \mb{x}^*\| \geq \chi \delta_p \quad \Longrightarrow \notag \\
    &\mathbb{E}
    [V(\mathscr{P}(\mb{x},d)) -V(\mb{x}) \mid \mb{x}]  \leq -k \|\mb{x}-\mb{x}^*\|^c.
    \label{eq: ExpLyapWalking}
\end{align}

This ISSp Lyapunov function can be utilized with the algorithmic approach introduced in \cite{tucker2023input} to solving the following optimization problem:
\begin{align}
    \label{eq: optdelta}
    (\delta_p^*,\chi_p^*) = \argmax_{\delta_p, \chi_p > 0} & ~ \delta_p \\
    \text{s.t. }  &  \mathbb{E}_{d \sim \mathcal{D}} [V(\mathscr{P}(\mb{x},d)) -V(\mb{x}) \mid \mb{x}]  \leq - k\|  \mb{x}   \|^2  \notag \\
                  & \quad \forall ~ \|\mb{x}\| = \chi_p \delta_p, \quad ~ d \sim \mathcal{N}(0,\delta_p^2, \pm3\delta_p), \notag
\end{align}
As shown in Fig. \ref{fig: walkingalg}, the algorithmic approach to the updated optimization \eqref{eq: optdelta} results in more reasonable estimates of the maximum tolerable step height for each of the two gaits considered in \cite{tucker2023input}.

\newsec{Probabilistic Guarantees for ISSp}
While relaxing the Lyapunov condition to the one in \eqref{eq: optdelta} yields more realistic estimates of $\delta^*$, this relaxed condition no longer satisfies the assumptions needed to be provably ISS. Instead, we can use probabilistic bounds to assert that the system is ISSp. 

To do this, we first need to approximate a reasonable estimate of the Lyapunov level set that bounds the evolution of the system after $K$ steps. 
Rearranging the Lyapunov condition \eqref{eq: ExpLyapWalking2} for the largest $\chi_p^*$ and $\delta_p^*$ identified by Opt. \eqref{eq: optdelta}, and using the fact that $ \mathcal{W} \triangleq \{\mb{x} \mid V(\mb{x}) \leq \kappa_2(\chi \delta)^c \}$ in Theorem 2 of \cite{tucker2023input}, we obtain the Lyapunov bound:
\begin{align*}
    V(\mb{x}_K) \leq \tilde{\rho} \triangleq (1-\alpha)^K \lambda_{\max}(P)(\chi_p^*\delta_p^*)^2 + k(\chi_p^* \delta_p^*)^2.
\end{align*}

As discussed in Sec. \ref{sec:LQG}, the associated probabilistic bound associated with remaining within this Lyapunov level set can be obtained from Kushner \cite{kushner_stochastic_1967}. Importantly, when your Lyapunov level set $\tilde{\rho}$ is less than $\tilde{\sigma}/\alpha$, the bound \eqref{eq:ville_lqg} is extremely conservative. Following Kushner, it is possible to find a better choice of $W_k$ that yields a better probability bound. Specifically, one can use the bound:
\begin{align}
    \mathbb{P}&(V(\mb{x}_K) \leq \tilde{\rho}), ~\forall k \leq K\} \nonumber\\
    & \geq
    \begin{cases}
        \frac{\Tilde{\rho} - V(\mb{x}_0)}{\Tilde{\rho}} \left( \frac{\Tilde{\rho} - \Tilde{\sigma} }{\Tilde{\rho}}  \right)^K, & \Tilde{\rho} \geq \frac{\Tilde{\sigma}}{\alpha} \\
        1 - \frac{V(\mb{x}_0)(1 - \alpha)^K + \Tilde{\sigma}\sum_{i=1}^K(1 - \alpha)^{i-1}}{\Tilde{\rho}} , & \textrm{otherwise.}
    \end{cases}
\end{align}
with $\tilde{\sigma}$ determined for each $\delta_p$ via Monte Carlo sampling. 

In Fig. \ref{fig: probbound2}, the probabilistic bound is illustrated for the values $(\chi_p^*,\delta_p^*)$ obtained using \eqref{eq: optdelta} for Gait 2. To verify the probabilistic bound, Monte Carlo sampling was implemented to estimate the true probability that the system remains within $\tilde{\rho}$ after $K=10$ steps. We simulate the system for this horizon and report both the fraction of trajectories remaining stable, as well as the fraction of trajectories remaining in the Lyapunov sublevel set $\tilde{\rho}.$ 


\section{Conclusion}
In this paper, we introduced the notion of input-to-state stability in probability (ISSp), which generalizes ISS for discrete-time systems with unbounded stochastic disturbances. We provided Lyapunov conditions for exponential ISSp, drew connections between ISS, ISSp, and traditional stability notions for stochastic systems, and provided simulation studies of ISSp systems (including an LQG system and bipedal walker) where we provide practical, probabilistic stability guarantees for systems subject to random disturbances. 

This work opens numerous directions for future work. In particular, while the martingale-based tools used in this paper require only a simple expected drift condition \eqref{eq:iss_lyap} on the Lyapunov function, martingale-based probability bounds are typically considered quite weak. Thus, for practitioners, an important question is how to choose Lyapunov function $V,$ and the supermartingale $W,$ to obtain the tightest bound possible (see also sum-of-squares-based approaches \cite{steinhardt_finite-time_2012, santoyo_barrier_2021}). It is also likely that tighter bounds may be obtained by exploiting particular structure in the disturbance distribution and dynamics.

Another direction of interest is to investigate how stability notions degrade when the system is subject to state uncertainty (as compared to the process noise considered in this work). Like for ISSp, we would expect these stability guarantees to hold only in probability; it remains an open question whether a similar property to the ``smooth degredation'' observed under ISS and ISSp can be found for systems with uncertain state. 

Lastly, the application of ISSp to bipedal locomotion motivates its use for studying the robustness of periodic walking gaits to uncertain terrain. Future work in this area includes applying the notion of ISSp to the gait synthesis framework to systematically generate nominal walking trajectories that have probabilistic guarantees of robustness for reasonable estimates of uncertain terrain.

\bibliography{citations,cosner}
\bibliographystyle{ieeetr}
\cp 

\appendix
\subsection{Lemma for Lyapunov Conditions}
\label{lemma:power_split}
We use the following lemma to prove the sufficiency of the ISSp Lyapunov conditions in Theorem \ref{thm:lyap}. 

\begin{lemma}
For $x_1, x_2, p > 0$, there exists $\zeta > 0$ such that $(x_1^p + x_2^p)^\frac{1}{p} \leq \zeta(x_1 + x_2).$
\end{lemma}

\begin{proof}

\textbf{Case 1:} Suppose $p \geq 1.$ Then $(\lvert x_1\rvert^p + \lvert x_2\rvert^p)^\frac{1}{p} \triangleq \norm{\mb{x}}_p$ defines the $\ell_p$ norm for $\mb{x} \triangleq [x_1, x_2]^T$ on $\mathbb{R}^2.$ Since $\ell_p$ norms are equivalent \cite{boyd_convex_2004}, there exists $\zeta > 0$ such that $\norm{\mb{x}}_p \leq \zeta \norm{\mb{x}}_1 = \zeta(\lvert x_1 \rvert + \lvert x_2 \rvert).$  The result follows since $x_1, x_2 > 0.$

\textbf{Case 2:} Suppose $0 < p < 1$. Then, $x^p$ is a concave function \cite{boyd_convex_2004}, and
$\left(\frac{x_1^p + x_2^p}{2} \right) \leq \left(\frac{x_1 + x_2}{2}\right)^p.$ Since $x^{\frac{1}{p}}$ is an increasing function, exponentiating both sides preserves ordering,
$\left(\frac{x_1^p + x_2^p}{2} \right)^\frac{1}{p} \leq \frac{x_1 + x_2}{2},$
and thus we have
$(x_1^p + x_2^p)^\frac{1}{p} \leq 2^{\frac{1}{p}-1} (x_1 + x_2) \triangleq \zeta (x_1 + x_2),
$
with $\zeta > 0$ as needed.
\end{proof}


\end{document}